\documentclass[conference, 10pt]{ieeeconf}
\IEEEoverridecommandlockouts     

\usepackage{cite}
\usepackage{comment}

\usepackage{amssymb,amsfonts,amsmath,amsthm,amscd}

\newtheorem{theorem}{Theorem}
\newtheorem{lemma}{Lemma}

\newtheorem{proposition}{Proposition}
\newtheorem{definition}{Definition}

\newcommand{\cB}{\ensuremath{\mathcal B}}
\newcommand{\cI}{\ensuremath{\mathcal I}}

\newcommand{\sU}{\ensuremath{\mathsf U}}

\newcommand{\sY}{\ensuremath{\mathsf Y}}

\newcommand{\PP}{\ensuremath{{\mathbb P}}}
\newcommand{\QQ}{\ensuremath{{\mathbb Q}}}
\newcommand{\E}{\ensuremath{\mathbb E}}

\newcommand{\crit}{\ensuremath{\underline{\sigma}}}

\newcommand{\Naturals}{\ensuremath{\mathbb N}}
\newcommand{\Reals}{\ensuremath{\mathbb R}}

\def\eps{\varepsilon}

\newcommand{\tr}{\ensuremath{{\scriptscriptstyle\mathsf{T}}}}

\newcommand{\deq}{\stackrel{\scriptscriptstyle\triangle}{=}}
\def\wh#1{\ensuremath{\hat{#1}}}
\def\td#1{\ensuremath{\tilde{#1}}}

\def\argmin{\operatornamewithlimits{arg\,min}}

\title{\Large \bf Divergence-Based Characterization of
Fundamental Limitations\\of Adaptive Dynamical Systems}
\author{Maxim Raginsky
\thanks{This work was supported by NSF grant CCF-1017564 and by AFOSR grant FA9550-10-1-0390.}
\thanks{The author is with the Department of Electrical and Computer Engineering, Duke University, Durham, NC. E-mail: {m.raginsky@duke.edu}.}}

\begin{document}

\maketitle

\begin{abstract}
	Adaptive dynamical systems arise in a multitude of contexts, e.g., optimization, control, communications, signal processing, and machine learning. A precise characterization of their fundamental limitations is therefore of paramount importance. In this paper, we consider the general problem of adaptively controlling and/or identifying a stochastic dynamical system, where our {\em a priori} knowledge allows us to place the system in a subset of a metric space (the uncertainty set). We present an information-theoretic meta-theorem that captures the trade-off between the metric complexity (or richness) of the uncertainty set, the amount of information acquired online in the process of controlling and observing the system, and the residual uncertainty remaining after the observations have been collected. Following the approach of Zames, we quantify {\em a priori} information by the Kolmogorov (metric) entropy of the uncertainty set, while the information acquired online is expressed as a sum of information divergences. The general theory is used to derive new minimax lower bounds on the metric identification error, as well as to give a simple derivation of the minimum time needed to stabilize an uncertain stochastic linear system.
\end{abstract}

\thispagestyle{empty}
\pagestyle{empty}

\section{Introduction}
\label{sec:intro}

What is adaptation? What is learning? These two questions arise all the time in practically any discussion of complex systems exhibiting complex behaviors. In control theory, these notions were a consistent theme in the work of George Zames (see, e.g.,~\cite{Zames} and references therein), who has put forward the following theses:
\begin{enumerate}
	\item Adaptation and learning involve acquisition of information about the object (system) being controlled.
	\item The appropriate notions of information are metric, locating the system in, say, a ball in a metric space.
	\item Acquiring information takes time.
	\item Nonadaptive (or robust) control optimizes performance on the basis of {\em a priori} information, whereas adaptive control is based on {\em a posteriori} information acquired online.
\end{enumerate}
In this paper, we take up the problem of characterizing the {\em fundamental limitations} of adaptive stochastic dynamical systems following the programme of Zames. We start by presenting a ``Meta-Theorem'' that ties together the three kinds of information mentioned by Zames: {\em a priori} information, represented by the metric complexity of the class of systems of interest; information acquired {\em online} as the system is being controlled; and {\em a posteriori} information, pertaining to the difficulty of identifying the system after a given length of time. Roughly speaking, given an arbitrary class of systems, an arbitrary controller, and an arbitrary identification algorithm, the Meta-Theorem quantifies the interplay and the trade-off between the initial uncertainty about the system, the online performance of the controller, and the final uncertainty remaining after the control task had been carried out.

We follow Zames in two key respects:
\begin{enumerate}
	\item We adopt the Kolmogorov entropy \cite{KolTih61} as our measure of {\em a priori} uncertainty (or complexity) of the class of systems at hand.
	\item We compare this initial uncertainty against the uncertainty remaining after the control signals have been applied.
\end{enumerate}
However, the novel aspect of our approach is the way in which we quantify the process of online information acquisition --- namely, through Shannon's information theory \cite{CovTho06}. Conceptually, our methodology is close to the way information-theoretic tools are being used in mathematical statistics to derive minimax bounds on the risk of statistical estimation procedures (see, e.g.,~\cite{Yu97,YanBar99,YangAISTATS} and references therein). The difference between statistical estimation and adaptive control, however, lies in the fact that, in control, we actively {\em intervene} into the system in order to steer it towards some desired state (control proper) or to learn something about the system (system identification). When we do not possess complete knowledge of the system, these two objectives may be in conflict, giving rise to the so-called {\em dual effect} of control \cite{BarTse}. With the exception of experimental design \cite{Fed72,Pan05} (and, in particular, some work connecting it with control \cite{LalleyLorden,GautierPronzato}), statistical estimation involves passively observing sample paths of a random process for the purpose of inference. Our Meta-Theorem covers both estimation and control, since the former can be viewed as an application of a control strategy that has no effect on the system, and it provides a way of quantifying the dual effect in the latter.

Following the statement and the proof of the Meta-Theorem in Section~\ref{sec:meta}, we show how it can be used to derive (a)  fundamental limits on the performance of system identification from input-output data, and (b) a lower bound on the minimum time needed to adaptively stabilize an uncertain linear system.

For system identification, we derive a minimax lower bound on the metric identification error, which shows that the intrinsic difficulty of identifying a system is determined by the balance of {\em a priori} metric information and the rate at which {\em a posteriori} information accumulates over time. We also show that ease of identification implies small {\em a priori} uncertainty. These results apply to any controller and any identification algorithm, providing yet another quantitative illustration of the dual effect. Bounds of similar flavor were derived by Yang \cite{YangAISTATS} in the context of statistical estimation from i.i.d.\ samples, and our techniques combine those of Yang with a more careful accounting of the accumulation of information during control/identification.

As for adaptive control, the first lower bounds on the rate of convergence in adaptive control are due to Nemirovski and Tsypkin \cite{NemTsy} (see also \cite{LaiCE} for further references), and we consider the same set-up. However, the proof in \cite{NemTsy} is rather lengthy and relies on the Cram\'er--Rao inequality. By contrast, we use the Meta-Theorem, which results in a much simpler and more direct information-theoretic argument.

\section{The ingredients: systems, controllers, identification algorithms}

A stochastic dynamical system is specified by a sequence of stochastic kernels relating present and past inputs and outputs to future outputs. The system is initially unknown, apart from the fact that we can place it in some {\em uncertainty set}, which is a subset of a metric space. The system is interconnected with a controller, which generates the inputs given past inputs and outputs. The exact purpose of control can be completely arbitrary, but we stipulate that the controller has been designed only with the knowledge of the uncertainty set. Finally, we consider the possibility that the observed temporal evolution of the system (i.e.,~its input-output trajectory) may be fed into an identification algorithm with the purpose of locating the system in a ``small'' region of the uncertainty set.

Specifically, we consider discrete-time stochastic dynamical systems with input space $\sU$ and output space $\sY$ (all spaces are assumed to be standard Borel \cite{BS}). The dynamics are assumed to be causal and nonanticipative, and so can be represented as a sequence of stochastic kernels $\{ P_\theta(dy_t|y^{t-1},u^{t-1}) \}^\infty_{t=1}$, where $\theta$ is a parameter that takes values in some metric space $(\Theta,\rho)$ and, for each $t$,
\begin{align}
	&\Pr\big(Y_t \in B \big| Y^{t-1} = y^{t-1}, U^{t-1} = u^{t-1}\big) \nonumber\\
	& \qquad \qquad \qquad \qquad = \int_B P(dy_t|y^{t-1},u^{t-1})
\end{align}
for every Borel set $B \subseteq \sY$. The inputs are generated by a controller, which is itself a dynamical system described by a sequence of stochastic kernels $\{ Q_\gamma(du_t | y^t,u^{t-1})\}^\infty_{t=1}$, where $\gamma$ is a parameter that takes values in some space $\Gamma$ that indexes the admissible controllers (e.g., open-loop, affine, Lipschitz, Markov, stationary, etc.). The system $\theta$ and the controller $\gamma$ are {\em interconnected} to form the joint probability law $\Pi_{\theta,\gamma}$ of $\{(Y_t,U_t)\}^\infty_{t=1}$ on $(\sY \times \sU)^\infty$, so that for each $T \in \Naturals$ we have
\begin{align}
	& \Pi_{\theta,\gamma}(dy^T, du^T) \nonumber\\
	& \qquad  = \bigotimes^{T}_{t=1} Q_\gamma(du_t|y^t,u^{t-1}) \otimes P_\theta(dy_{t}|y^{t-1},u^{t-1}).
\end{align}
Finally, we consider {\em identification algorithms} that observe the system trajectory $(Y_1,U_1),(Y_2,U_2),\ldots$ and attempt to estimate the true system model $\theta$. We will consider deterministic identification algorithms, so for each $T$ we define a $T$-step identification algorithm as a measurable mapping $\wh{\theta}_T : \sY^T \times \sU^T \to \Theta$.

\section{Prelude: identification error and metric complexity}

As stated earlier, we assume some {\em a priori} knowledge about the system of interest, namely that it lies in some uncertainty set $\Lambda \subseteq \Theta$. Since our primary interest is in capturing the interplay between identification and control, we need to quantify the extent to which the systems in $\Lambda$ can be identified after having been interconnected with a given controller $\gamma$ from $t=1$ to $t=T$:

\begin{definition} Consider a subset $\Lambda \subseteq \Theta$ of system models and a controller $\gamma$. Then the {\em $T$-step minimax identification error} on $\Lambda$ relative to $\gamma$ is given by
\begin{align}
	e_T(\Lambda,\gamma) \deq \inf_{\wh{\theta}_T} \sup_{\theta \in \Lambda} \E_{\theta,\gamma} \left\{ \rho\left(\wh{\theta}_T(Y^T,U^T),\theta\right)\right\},
\end{align}
where the infimum is over all $T$-step identification algorithms.
\end{definition}

\noindent The fact that the minimax identification error depends not only on the uncertainty set $\Lambda$, but also on the choice of the controller $\gamma$, is of key importance. The dependence on $\Lambda$ expresses the fact that some classes of systems are intrinsically more difficult to identify than others; the dependence on $\gamma$ captures the potential tension between control and identification/learning (the dual effect \cite{BarTse}). When system identification is the sole purpose, the controller $\gamma$ is typically open-loop \cite{Ljung78, Zames}, and the underlying deterministic sequence of inputs is chosen based on some criteria related to the structure of the uncertainty set, as well as to other constraints (e.g., stability, power, cost, etc.). However, there are also adaptive control strategies that adjust the behavior of the controller dynamically based on parameters estimated online \cite{GoodwinSin,KumarVaraiya}, and our definition of $e_T(\Lambda,\gamma)$ covers this possibility.

The basic idea, which in the context of control originated with Zames, is that the difficulty of identification is bound up with the richness of the uncertainty set $\Lambda$ --- the larger the uncertainty set, the harder it is to identify the system. We will combine this intuition with a probabilistic argument to show that, in a certain sense, system identification is no easier than hypothesis testing. Arguments of this sort are quite common in statistics \cite{Yu97,YanBar99}, but, as we shall see, they are equally applicable to control as well. To get things going, we start by proving a simple lower bound on $e_T(\Lambda,\gamma)$:

\begin{proposition}\label{prop:basic} Let $S$ be any finite \emph{$\eps$-separated subset of $\Lambda$}, i.e., for $S = \{\theta_1,\ldots,\theta_N\}$
	\begin{align}
		\rho(\theta_i,\theta_j) \ge \eps, \qquad \forall i \neq j.
	\end{align}
	Let $\cI_T(S)$ denote the set of all $T$-step identification algorithms taking values in $S$, i.e., $\cI_T(S) = \{ \wh{\theta}_T : \sY^T \times \sU^T \to S \}$.
	Then
	\begin{align}\label{eq:basic_minimax_bound}
		e_T(\Lambda,\gamma) \ge \frac{\eps}{2} \inf_{\wh{\theta}_T \in \cI_T(S)} \max_{\theta \in S} \Pi_{\theta,\gamma}\left\{ \wh{\theta}_T(Y^T,U^T) \neq \theta \right\}.
	\end{align}
\end{proposition}

\begin{proof} Using the fact that $S \subset \Lambda$ and Markov's inequality, we can write
	\begin{align}\label{eq:Markov_step}
		e_T(\Lambda,\gamma) \ge \frac{\eps}{2}\inf_{\wh{\theta}_T}\max_{\theta \in S} \Pi_{\theta,\gamma}\left\{ \rho(\wh{\theta}_T,\theta) \ge \eps/2 \right\}.
	\end{align}
	Given an arbitrary $\wh{\theta}_T$, define 
	\begin{align}
		\td{\theta}_T \deq \argmin_{\theta' \in S}\rho(\wh{\theta}_T,\theta').
	\end{align}
	Clearly, $\td{\theta}_T \in \cI_T(S)$. Suppose $\theta \in S$. If $\rho(\wh{\theta}_T,\theta) < \eps/2$, then necessarily $\rho(\wh{\theta}_T,\td{\theta}_T) < \eps/2$. If $\td{\theta}_T \neq \theta$, the triangle inequality gives
	\begin{align}
		\rho(\td{\theta}_T,\wh{\theta}_T) \ge \rho(\td{\theta}_T,\theta) - \rho(\wh{\theta}_T,\theta) \ge \eps/2,
	\end{align}
	which is a contradiction. Hence, if $\td{\theta}_T \neq \theta$, then $\rho(\wh{\theta}_T,\theta) \ge \eps/2$. Thus,
	\begin{align}
	&	\max_{\theta \in S}\Pi_{\theta,\gamma}\left\{ \rho(\wh{\theta}_T,\theta) \ge \eps/2 \right\} \nonumber\\
	& \qquad \ge \max_{\theta \in S}\Pi_{\theta,\gamma} \left\{ \td{\theta}_T \neq \theta \right\} \\
	& \qquad \ge \inf_{\wh{\theta}_T \in \cI_T(S)} \max_{\theta \in S} \Pi_{\theta,\gamma} \left\{ \wh{\theta}_T \neq \theta \right\}.
	\end{align}
	Combining this with \eqref{eq:Markov_step}, we get \eqref{eq:basic_minimax_bound}.
\end{proof}

The above proposition suggests a trade-off between the separation $\eps$ and the probability of correct identification. Indeed, if we make $\eps$ small, then the size of the maximal $\eps$-separated subset will be large, which in turn will tend to increase the probability of identification error. This observation naturally prompts us to take a look at the growth of maximal separated subsets of $\Lambda$ as a function of the separation $\eps$, which is captured by Kolmogorov's notion of the metric entropy \cite{KolTih61}:

\begin{definition} Given a set $\Lambda \subseteq \Theta$, we define its {\em packing numbers} by
	\begin{align}
	& 	N_\rho(\eps;\Lambda) \deq \max\Big\{ N \ge 1 :\nonumber \\
	& \qquad	 \exists \theta_1,\ldots,\theta_N \in \Lambda \text{ \rm s.t. } \rho(\theta_i,\theta_j) \ge \eps, \forall i \neq j \Big\}
	\end{align}
and the corresponding Kolmogorov entropy by $H_\rho(\eps; \Lambda) \deq \log N_\rho(\eps;\Lambda)$.
\end{definition}

\section{The Meta-Theorem}
\label{sec:meta}

Now that all the ingredients are in place, we can state and prove our Meta-Theorem, which captures the interplay between the metric complexity of the uncertainty set $\Lambda$ ({\em a priori} information, as per Zames), the information acquired online by acting on the system and observing its response, and the uncertainty remaining after $T$ time steps. The main idea is to embed the problem of adaptive control and identification in a ``doubly stochastic'' set-up, in which Nature first selects a system at random from an $\eps$-separated subset of $\Lambda$, and then this system is interconnected with a given controller and fed into a given identification algorithm. The Meta-Theorem applies to any uncertainty set, any controller, and any identification algorithm. Our usage of the prefix ``meta'' is intended to draw parallels to recent work of Polyanskiy et al.~\cite{PPV,PolThesis}, which develops a ``meta-converse'' for channel coding by relating the performance of any channel coding scheme on one channel to its performance on another (we will elaborate on these parallels shortly).

Given a separation $\eps > 0$, let $\Lambda_\eps = \{\theta_1,\ldots,\theta_N\} \subset \Lambda$, $N = N_\rho(\eps; \Lambda)$, be any maximal $\eps$-packing set, and suppose that the system model is drawn {\em uniformly at random} from $\Lambda_\eps$. Then this system is interconnected with a given controller $\gamma$. To describe all the events pertaining to this interconnection, we construct a probability space $(\Omega,\cB,\PP)$ with the following random variables defined on it:
\begin{itemize}
	\item $W \in [N]$, the random choice of a system model in $\Lambda_\eps$
	\item $U^T \in \sU^T$, the inputs applied to the system by $\gamma$
	\item $Y^T \in \sY^T$, the resulting outputs.
\end{itemize}
These variables describe the interaction between the system and the controller, and thus have the causal ordering
\begin{align}
	W,Y_1,U_1,\ldots,Y_t,U_t,\ldots,Y_T,U_T,
\end{align}
where, $\PP$-almost surely,
\begin{align}
	& \PP(W = i) = \frac{1}{N}, \forall i \in [N] \\
	& \PP(U_t \in A | W,Y^t,U^{t-1}) = Q_{\gamma}(A | Y^t,U^{t-1}) \\
	& \PP(Y_t \in B | W,Y^{t-1},U^{t-1}) = P_{\theta_W}(B | Y^{t-1},U^{t-1})
\end{align}
for all Borel sets $A \subseteq \sU, B \subseteq \sY$. In other words, $W \to (Y^t,U^{t-1}) \to U_t$ is a Markov chain for each $t$. To simplify notation, let us denote by $Z_t$ the pair $(Y_t,U_t)$. At time $T$ the entire sequence $Z^T = (Z_1,\ldots,Z_T)$ is fed into an identification algorithm $\wh{\theta}_T$.

With these definitions, we are now in a position to state the Meta-Theorem:

\begin{theorem}\label{thm:main} Consider any controller $\gamma$ and any $T$-step identification algorithm $\wh{\theta}_T \in \cI_T(\Lambda_\eps)$. Then the bound
	\begin{align}\label{eq:main_bound}
 &H_\rho(\eps; \Lambda) \cdot \min_{\theta \in \Lambda_\eps} \Pi_{\theta,\gamma} \left\{ \wh{\theta}_T = \theta \right\}\nonumber \nonumber\\
& \le \sum^T_{t=1} D\big( \PP_{Y_t|Z^{t-1},W}\big\| \QQ_{Y_t|Z^{t-1}} \big| \PP_{U_t,Z^{t-1},W} \big) + \log 2
	\end{align}
holds for any sequence of stochastic kernels $\{ \QQ_{Y_t|Z^{t-1}}\}^T_{t=1}$ that satisfy the condition $\PP_{Y_t|Z^{t-1}} \ll \QQ_{Y_t|Z^{t-1}}, \forall t$.
\end{theorem}

\begin{proof} We start by observing that
	\begin{align}
		\max_{\theta \in \Lambda_\eps} \Pi_{\theta,\gamma} \left\{ \wh{\theta}_T \neq \theta \right\} \ge \inf_{\wh{W}}\PP \left\{ \wh{W} \neq W \right\},
	\end{align}
	where the infimum is over all estimators $\wh{W} : \sY^T \times \sU^T \to [N]$. Since any such $\wh{W}$ is $\sigma(Z^T)$-measurable and since $W$ is uniformly distributed on $[N]$, we can apply Fano's inequality \cite{CovTho06,HanVer94} to write
	\begin{align}
		\inf_{\wh{W}} \PP\{ \wh{W} \neq W \} \ge 1 - \frac{I(W; Z^T) + \log 2}{\log N},
	\end{align}
	where $I(W; Z^T)$ is the mutual information between $W$ and $Z^T = (Y^T,U^T)$ under $\PP$. We now expand this mutual information:
	\begin{align}
		&I(W; Z^T) = \sum^T_{t=1} I(W; Z_t | Z^{t-1}) \\
		&\quad = \sum^T_{t=1} I(W; Y_t,U_t|Z^{t-1}) \\
		&\quad = \sum^T_{t=1} [I(W; Y_t | Z^{t-1}) + I(W; U_t | Y_t, Z^{t-1})] \\
		&\quad = \sum^T_{t=1} I(W; Y_t | Z^{t-1}),\label{eq:info_sum}
	\end{align}
where the first three steps follow from the repeated application of the chain rule, while the last step uses the fact that $W \to (Y_t,Z^{t-1}) \to U_t$ is a Markov chain. Now, for each summand in \eqref{eq:info_sum} we have
\begin{align}
	& I(W; Y_t | Z^{t-1}) \nonumber\\
	&= D\big( \PP_{Y_t|Z^{t-1},W} \big\| \PP_{Y_t|Z^{t-1}} \big| \PP_{Z^{t-1},W}\big) \\
	&= \E \left\{ \log \frac{d\PP_{Y_t|Z^{t-1},W}}{d\PP_{Y_t|Z^{t-1}}} \right\} \\
	&= \E \left\{ \log \frac{d\PP_{Y_t|Z^{t-1},W}}{d\QQ_{Y_t|Z^{t-1}}} \right\} - \E \left\{ \log \frac{d\PP_{Y_t|Z^{t-1}}}{d\QQ_{Y_t|Z^{t-1}}}\right\} \\
	&= D\big( \PP_{Y_t|Z^{t-1},W} \big\| \QQ_{Y_t|Z^{t-1}} \big| \PP_{Z^{t-1},W}\big) \nonumber\\
	& \qquad \qquad - D\big( \PP_{Y_t|Z^{t-1}} \big\| \QQ_{Y_t|Z^{t-1}} \big| \PP_{Z^{t-1}}\big) \\
	& \le D\big( \PP_{Y_t|Z^{t-1},W} \big\| \QQ_{Y_t|Z^{t-1}} \big| \PP_{Z^{t-1},W}\big),
\end{align}
where the first two steps use the definition of conditional mutual information, the next step follows from the fact that $\PP_{Y_t|Z^{t-1}} \ll \QQ_{Y_t|Z^{t-1}}$ for every $t$, the step after that uses the definition of conditional divergence, and the last step follows because the divergence is nonnegative. Combining everything, we obtain the desired bound \eqref{eq:main_bound}.
\end{proof}

Note that the left-hand side of \eqref{eq:main_bound} involves the initial amount of uncertainty about the system (the metric entropy) and the best identification error performance at time $T$, while the right-hand side is a sum of information divergences added up from $t=1$ to $t=T$. The main power of the Meta-Theorem resides in the freedom to choose the auxiliary stochastic kernels $\{\QQ_{Y_t|Z^{t-1}}\}^T_{t=1}$. For example, we may consider the case in which $\gamma$ is designed for some ``nominal'' system $\theta_0 \in \Theta$, and we can take $\QQ_{Y_t|Z^{t-1}}$ to be the transition law of $\theta_0$ controlled by $\gamma$. With this choice, the $t$th term on the right-hand side of \eqref{eq:main_bound} quantifies the ``robustness radius'' of $\gamma$ on $\Lambda$ at time $t$. Alternatively, we may consider the setting, in which there is an optimal controller $\gamma_\theta$ associated to each $\theta \in \Theta$, and
\begin{align}
	\Pi_{\theta,\gamma_\theta}(dY_t|Z^{t-1}) = \Pi_{\theta',\gamma_{\theta'}}(dY_t|Z^{t-1})
\end{align}
for all $\theta,\theta' \in \Lambda$. In that case, we may take $\QQ_{Y_t|Z^{t-1}}$ to be the controlled transition law of $\theta$ interconnected with $\gamma_\theta$ (for any $\theta$). With this choice, the $t$th term on the right-hand side of \eqref{eq:main_bound} tells us by how much the actual performance of $\gamma$ operating in the presence of uncertainty differs from that of the optimal controller at time $t$ when there is no uncertainty. In general, the use of an auxiliary sequence of $\QQ$-kernels is similar to the use of auxiliary channels in the information-theoretic ``meta-converse'' of Polyanskiy et al.~\cite{PPV,PolThesis}.

The remainder of the paper is devoted to several sample applications of the Meta-Theorem, intended to showcase its power and flexibility.

\section{Fundamental limits of identification}

Our first application of the Meta-Theorem concerns the fundamental limitations of system identification algorithms. For the results of this section, the precise structure of the controller $\gamma$ is irrelevant, and the influence of $\gamma$ manifests itself indirectly through time-dependent bounds on the metric identification error. For notational simplicity, we will denote by $P_{\theta,t}$ the stochastic kernel $P_\theta(dy_t|y^{t-1},u^{t-1})$, where it is understood that $P_{\theta,t}$ is a Borel probability measure on $\sY$ and a Borel-measurable function of $(y^{t-1},u^{t-1})$.

The nature of the results presented below, and the techniques used to prove them, are inspired by the work of Yang \cite{YangAISTATS} on the limits of regression learning procedures in statistics. Moreover, the statistical estimation setting is subsumed by our results since a stochastic process with sample paths in $\sY^\infty$ and with parameter $\theta \in \Theta$ can be viewed as a dynamical system $\{P_\theta(dy_t|y^{t-1})\}^\infty_{t=1}$ (i.e.,~the controller does not affect the system).

\subsection{The Critical Separation bound}

The first result we prove is a lower bound on the $T$-step minimax identification error, which is expressed in terms of {\em upper} bounds for a sequence of $t$-step identification algorithms, from $t=0$ (i.e., any data-free guess about the system parameter $\theta$) to $t=T-1$:

\begin{theorem}\label{thm:ID_lower_bound} Consider a model class $\Lambda$ and a controller $\gamma$. Suppose that there exists a sequence $\{ \wh{\theta}_t \}^{T-1}_{t=0}$ of identification algorithms, such that
	\begin{align}
		\sup_{\theta \in \Lambda} \E_{\theta,\gamma} D\Big(P_{\theta,t} \Big\| P_{\wh{\theta}_{t-1},t}\Big) \le \delta_t, \qquad \forall t.
	\end{align}
	Then
	\begin{align}\label{eq:critical_minimax}
		e_T(\Lambda,\gamma) \ge \frac{\crit_{T}}{4},
	\end{align}
	where the {\em critical separation} $\crit_T$ is chosen so that
	\begin{align}\label{eq:critical_radius}
		H_\rho(\crit_T; \Lambda) = \left\lceil 2\left(\sum^T_{t=1} \delta_t + \log 2\right) \right\rceil.
	\end{align}
\end{theorem}

\begin{proof} Consider the setting of Theorem~\ref{thm:main} with the given $\Lambda,\gamma$ and $\eps = \crit_T$ defined according to \eqref{eq:critical_radius}. For each $t$, let $\QQ_{Y_t|Z^{t-1}}$ be defined via
	\begin{align}
		\QQ(Y_t \in B | Z^{t-1}) = P_{\wh{\theta}_{t-1}(Z^{t-1})}(B | Z^{t-1})
	\end{align}
	for any Borel set $B \subseteq \sY$. Then
	\begin{align}
		& D \big( \PP_{Y_t|Z^{t-1},W}\big\| \QQ_{Y_t|Z^{t-1}} \big| \PP_{Z^{t-1},W} \big) \nonumber\\
		& = \frac{1}{N}\sum^N_{i=1} \int \PP(dz^{t-1}|W=i) D\Big(P_{\theta_i,t}\Big\| P_{\wh{\theta}_{t-1}(z^{t-1}),t}\Big) \\
		& \le \sup_{\theta \in \Lambda} \int \Pi_{\theta,\gamma}(dz^{t-1}) D\Big( P_{\theta,t} \Big\| P_{\wh{\theta}_{t-1}(z^{t-1}),t}\Big) \\
		&= \sup_{\theta \in \Lambda} \E_{\theta,\gamma} D\Big( P_{\theta,t} \Big\| P_{\wh{\theta}_{t-1}(z^{t-1}),t}\Big) \\
		&\le \delta_t.
	\end{align}
	Then, for any $\wh{\theta}_T$ taking values in $\Lambda_{\crit_T}$,
	\begin{align}
		H_\rho(\crit_T; \Lambda) \min_{\theta \in \Lambda_{\crit_T}} \Pi_{\theta,\gamma}\left\{ \wh{\theta}_T = \theta \right\} \le \sum^T_{t=1}\delta_t + \log 2.
	\end{align}
Combining this with \eqref{eq:critical_radius} and noting that $\wh{\theta}_T$ was arbitrary, we get
\begin{align}
	\inf_{\wh{\theta}_T \in \cI_T(\Lambda_{\crit_t})}\max_{\theta \in \Lambda_{\crit_T}} \Pi_{\theta,\gamma}\left\{ \wh{\theta}_T \neq \theta \right\} \ge \frac{1}{2}.
\end{align}
Finally, substituting this into the lower bound \eqref{eq:basic_minimax_bound}, we get \eqref{eq:critical_minimax}.
\end{proof}

\subsection{Easy identification implies small {\em a priori} uncertainty}

We now use Theorem~\ref{thm:ID_lower_bound} to prove that any class of systems that are easy to identify (in the sense that there exists a sequence of identification algorithms whose worst-case errors over the class decay at some prescribed rate) must necessarily have correspondingly small metric entropy. In other words, if a class of systems is easy to identify, then its {\em a priori} uncertainty could not have been very large.

To formalize things, consider a controller $\gamma$, a sequence of identification schemes $\{ \wh{\theta}_t\}^\infty_{t=0}$, and a nonincreasing sequence of positive reals $\{\beta_t\}^\infty_{t=0}$. For a given $k \ge 1$, let us define the set $\Lambda_k( \gamma, \{ \wh{\theta}_t\}^\infty_{t=0}, \{ \beta_t\}^\infty_{t=0})$ to consist of all systems $\theta \in \Lambda$, such that
\begin{align}
	\E_{\theta,\gamma} \rho^k(\wh{\theta}_{t},\theta) \le \beta_t, \qquad \forall t.
\end{align}

\begin{theorem} Suppose that $\gamma$ is such that, for all $t$ and all $\theta,\theta' \in \Theta$,
	\begin{align}\label{eq:smoothness}
		\E_{\theta,\gamma}D \big( P_{\theta,t} \big\| P_{\theta',t}\big) \le K  \rho^k(\theta,\theta')
	\end{align}
	for some $K > 0$. Then the class $\Lambda = \Lambda_k(\gamma, \{\wh{\theta}\}_t, \{\beta_t\})$ satisfies the bound
	\begin{align}
		H_\rho\left(5\beta_T^{1/k}; \Lambda\right) \le \left\lceil 2 \left( K\sum^T_{t=1} \beta_{t-1} + \log 2 \right)\right\rceil
	\end{align}
	for every $T$.
\end{theorem}

\begin{proof} From the smoothness condition \eqref{eq:smoothness} it follows that
	\begin{align}
		\E_{\theta,\gamma} D\big( P_{\theta,t} \big\| P_{\wh{\theta}_{t-1},t} \big) \le K \beta_{t-1}
	\end{align}
	for every $t \ge 1$. Hence, applying Theorem~\ref{thm:ID_lower_bound} with $\delta_t = K \beta_{t-1}$ we get
	\begin{align}
		e_T(\Lambda,\gamma) \ge \frac{\crit_T}{4},
	\end{align}
	where $\crit_T$ is chosen according to \eqref{eq:critical_radius}:
	\begin{align}\label{eq:crit}
		H_\rho(\crit_T; \Lambda) = \left\lceil 2\left( K \sum^T_{T=1} \beta_{t-1} + \log 2\right)\right\rceil.
	\end{align}
	Let $H_T$ denote the quantity on the right-hand side of \eqref{eq:crit}. Let us suppose that $H_\rho\left(5\beta_T^{1/k}; \Lambda\right) > H_T$. Then, because the mapping $\eps \mapsto H_\rho(\eps;\Lambda)$ is monotone decreasing, we must have $5\beta_T^{1/k} \le \crit_T$. But that implies that
	\begin{align}\label{eq:entropy_lower}
		e_T(\Lambda,\gamma) \ge \frac{\crit_T}{4} \ge \frac{5\beta_T^{1/k}}{4} > \beta_T^{1/k}.
	\end{align}
	On the other hand, for any $\theta \in \Lambda$ we have
	\begin{align}
		\E_{\theta,\gamma}\rho(\wh{\theta}_t,\theta) \le \left( \E_{\theta,\gamma}\rho^k(\wh{\theta}_t,\theta)\right)^{1/k} \le \beta^{1/k}_t,
	\end{align}
	where the first step uses Jensen's inequality and the second step uses the definition of $\Lambda$. This implies, in turn, that
	\begin{align}
		e_T(\Lambda,\gamma) \le \E_{\theta,\gamma} \rho(\wh{\theta}_T,\theta) \le \beta_T^{1/k},
	\end{align}
	which contradicts \eqref{eq:entropy_lower}. Hence, $H_\rho\left(5\beta^{1/k}_T; \Lambda\right) \le H_T$.
\end{proof}

As an example of when the smoothness condition \eqref{eq:smoothness} holds, consider a first-order nonlinear system of the form
	\begin{align}
		Y_t = f_\theta(Y_{t-1}) + U_{t-1} + V_t,
	\end{align}
	where $\sY = \sU = \Reals$ and $\{V_t\}$ is an i.i.d.\ sequence of Gaussian random variables with zero mean and variance $\sigma^2$. Suppose that the mappings $f_\theta$ satisfy the condition
	\begin{align}
		|f_\theta(y) - f_{\theta'}(y)|^2 \le K_0 F(y) \rho^k(\theta,\theta'), \quad \forall \theta,\theta' \in \Theta
	\end{align}
	for some $K_0 > 0$, $k \ge 1$, and some function $F : \Reals \to \Reals$ which is bounded on compacts. Then, provided $\gamma$ is chosen so that there exists some finite $R > 0$, such that  $|Y_t| \le R$ $\Pi_{\theta,\gamma}$-almost surely for every $\theta \in \Theta$, we will have, for any $\theta,\theta' \in \Theta$
	\begin{align}
		\E_{\theta,\gamma} D(P_{\theta,t}\|P_{\theta',t}) & = \frac{1}{2\sigma^2} \E_{\theta,\gamma}|f_\theta(Y_t) - f_{\theta'}(Y_t)|^2\\
		& \le \frac{K_0}{2\sigma^2} \max_{|y| \le R} F(y) \cdot \rho^k(\theta,\theta').
	\end{align}
	
To appreciate the implications of the above result, we can consider the following cases:
\begin{enumerate}
	\item $\beta_t \le Ct^{-\alpha}$ for some $C > 0$ and $0 < \alpha < 1$. Then, for all sufficiently small $\eps$, we will have
	\begin{align}
		H_\rho(\eps; \Lambda) \le C' \left(\frac{1}{\eps}\right)^{\frac{2(1-\alpha)}{k\alpha}},
	\end{align}
	where $C' > 0$ is a constant that depends only on $K, k, \alpha, C$. In this case, the metric complexity of $\Lambda$ is, essentially, that of a ball in an infinite-dimensional Hilbert space.
	\item $\beta_t \le Ct^{-1}$ for some $C > 0$. Then, for all sufficiently small $\eps$, we will have
	\begin{align}
		H_\rho(\eps; \Lambda) \le C' k \log \frac{1}{\eps},
	\end{align}
	where $C' > 0$ is a constant that depends only on $K, k, C$. In this case, $\Lambda$ is, essentially, a ball in a finite-dimensional Hilbert space.
\end{enumerate}

\section{Rates of convergence in adaptive control}

In this section, we will use the Meta-Theorem to derive a fundamental limit on the minimum time needed to achieve a particular control objective.

Consider the problem of adaptively controlling a first-order $n$-dimensional linear system
\begin{align}
Y_{t+1} = AY_t + U_t + V_{t+1}, \qquad t=1,2,\ldots
\end{align}
where $\sU = \sY = \Reals^n$, $\{U_t\}^\infty_{t=1}$ is the input (control) sequence, $\{Y_t\}^\infty_{t=1}$ is the output sequence, and $\{V_t\}^\infty_{t=1}$ is an i.i.d.\ Gaussian disturbance process with zero mean and covariance matrix $\sigma^2 I_{n \times n}$, independent of the initial state $Y_1$. We assume that the initial state $Y_1$ has a finite second moment, $\E \| Y_1 \|^2 = C < \infty$.  The unknown system matrix $A \in \Reals^{n \times n}$ is assumed to lie in the set
\begin{align}
	\Lambda = \{ A \in \Reals^{n \times n} : \| A \| \le 1 \},
\end{align}
where $\| \cdot \|$ denotes the operator (spectral) norm. The space of admissible controllers $\Gamma$ is assumed to consist of sequences $\gamma = \{\gamma_t\}^\infty_{t=1}$ of deterministic Borel mappings $\gamma_t : \sY^t \times \sU^{t-1} \to \sU$, so that $U_t = \gamma_t(Y^t,U^{t-1})$. The objective is to select a control law $\gamma^* \in \Gamma$ such that
\begin{align}\label{eq:objective}
& \limsup_{T \to \infty}  \E_{A,\gamma^*} \left\{\frac{1}{T} \sum^{T}_{t=1} \| Y_{t+1}\|^2 \right\} \nonumber \\
& \qquad = \inf_{\gamma\in\Gamma} \limsup_{T \to \infty} \E_{A,\gamma} \left\{ \frac{1}{T}\sum^{T}_{t=1} \|Y_{t+1}\|^2\right\}
\end{align}
for every $A \in \Lambda$.

Following Lai \cite{Lai86}, we can define the $T$-step {\em regret} of $\gamma$ on $A$ by
\begin{align}
R_T(\gamma,A) \deq \E_{A,\gamma} \left\{\sum^{T}_{t=1} \| Y_{t+1} - V_{t+1} \|^2 \right\}.
\end{align}
Since $Y_{t+1} - V_{t+1}$ is independent of $V_{t+1}$, we can write
\begin{align}
\E \| Y_{t+1} \|^2 &= \E \| Y_{t+1} - V_{t+1} \|^2 + n\sigma^2 \\
&= \E \| AY_t + U_t \|^2 + n\sigma^2 \\
&\ge n\sigma^2.
\end{align}
This implies that the the infimum on the right-hand side of \eqref{eq:objective} is equal to $n\sigma^2$; consequently, we seek a $\gamma^*$ such that, for all $A \in \Lambda$,
\begin{align}
\limsup_{T \to \infty} \frac{R_T(\gamma^*,A)}{T} = \inf_{\gamma} \limsup_{T \to \infty} \frac{R_T(\gamma,A)}{T} = 0.
\end{align}
Lai \cite{Lai86} calls any such $\gamma^*$ {\em asymptotically efficient}.

Given a controller $\gamma \in \Gamma$, let us define the quantity
\begin{align}
T^*_\gamma(\eps) \deq \sup_{A \in \Lambda} \inf \left\{ T \ge 1 : \frac{R_T(\gamma,A)}{T} < \eps \right\}
\end{align}
This is the minimum time it takes $\gamma$ to achieve average regret of less than $\eps$ on every $A \in \Lambda$. We will obtain a lower bound on $T^*_\gamma(\eps)$ for any $\gamma$ that has a certain property known as {\em persistent excitation} (cf.~\cite{Lai86,LaiCE,KumarVaraiya,Duflo97}):

\begin{definition} Given $c > 0$ and $\delta \in (0,1)$, a controller $\gamma \in \Gamma$ has the \emph{$(c,\delta)$-persistent excitation property} if there exists some $T_0 \in \Naturals$ such that, for every $A \in \Lambda$,
	\begin{align}\label{eq:PEP}
		\Pi_{A,\gamma} \left( \frac{1}{T} \sum^T_{t=1}Y_t Y^\intercal_t \succeq c I_{n \times n} \right) \ge 1-\delta, \quad \forall T \ge T_0
	\end{align}
	where for any two $M_1,M_2 \in \Reals^{n \times n}$ the notation $M_1 \succeq M_2$ means that $M_1 - M_2$ is a positive semidefinite matrix.
\end{definition}
\noindent Our main result is as follows:

\begin{theorem} Any controller $\gamma \in \Gamma$ that has the $(c,\delta)$-persistent excitation property with $\delta < 1/4$ must satisfy
\begin{align}\label{eq:mintime}
T^*_\gamma(\eps) = \Omega\left(\frac{n^2 \sigma^2}{\eps} \log \frac{1}{\eps}\right),
\end{align}
where the constant implicit in the $\Omega(\cdot)$ notation depends only on $c$ and $\delta$.
\end{theorem}

\begin{proof}
	
We first show that any good controller can be used to construct a good identification scheme. The proof of this assertion essentially follows Nemirovski and Tsypkin \cite{NemTsy}.

Given a controller $\gamma = \{\gamma_t\}$, we first note that the probability that any component of $Y_t$ vanishes is zero. Hence, without loss of generality for every $t$ we can write 
\begin{align}
\gamma_t(Y^t,U^{t-1}) = - F_t(Y^t,U^{t-1})Y_t, \qquad \text{a.s.}
\end{align}
for some measurable mapping $F_t : \sY^t \times \sU^{t-1}\to \Reals^{n \times n}$. Now for each $T$ let
\begin{align}
	G_T \deq \sum^{T}_{t=1} Y_t Y_t^\tr
	\end{align}
	and consider the following least-squares identification algorithm:
\begin{align}
\tilde{A}_T \deq \begin{cases}
0, & \text{if } \det G_T = 0 \\
\displaystyle \sum^{T}_{t=1} F_t(Y^t,U^{t-1}) Y_t Y_t^\tr G_{T}^{-1}, & \text{otherwise}
\end{cases}
\end{align}
For this identification algorithm, we have the following lemma, whose proof is presented in Appendix~\ref{app:ident_proof}:

\begin{lemma}\label{lm:control_ident} Suppose $\gamma$ has the $(c,\delta)$-persistent excitation property. Then for every $A \in \Lambda$ and for every $T \ge T_0$,
\begin{align}
\| \tilde{A}_T - A \|^2 \le \frac{1}{cT}\sum^{T}_{t=1} \| Y_{t+1} - V_{t+1} \|^2
\end{align}
 with $\Pi_{A,\gamma}$-probability at least $1-\delta$.
\end{lemma}

Next we show that if $\gamma$ achieves average regret of less than $\eps$ in $T$ time steps, then the corresponding identification scheme $\td{A}_T$ must have a small probability of error.

Given $\eps$, let $N_{\| \cdot \|}(\eps; \Lambda)$ denote the $\eps$-packing number of $\Lambda$ w.r.t.\ the metric induced by the spectral norm. Since $\Lambda$ is a norm ball in $\Reals^{n^2}$, there exist constants $b_n,c_n > 0$, such that
\begin{align}\label{eq:packing}
b_n + n^2 \log \frac{1}{\eps} \le H_{\| \cdot \|}(\eps; \Lambda) \le c_n + n^2 \log \frac{1}{\eps}
\end{align}
for all sufficiently small $\eps > 0$. Now let $N(\eps) = N_{\| \cdot \|}(4\sqrt{\eps/c}; \Lambda)$ and take $\{A_1,\ldots,A_N\} \subset \Lambda$ to be a maximal $4\sqrt{\eps/c}$-packing set. Given a controller $\gamma$, define
\begin{align}\label{eq:estimator}
\wh{W} \deq \argmin_{1 \le i \le N(\eps)} \| \tilde{A}_T - A_i \|.
\end{align}
Then we have the following lemma, whose proof is given in Appendix~\ref{app:upper_proof}:

\begin{lemma}\label{lm:prob_error_upper} Suppose that $\gamma$ has the $(c,\delta)$-persistent excitation property and achieves regret $< \eps$ in time $T$. Let $W$ be a random variable uniformly distributed over the set $\{1,\ldots,N(\eps)\}$ independently of $Y_1,\{V_t\}$. Then the estimator \eqref{eq:estimator} satisfies
\begin{align}
\PP\left( \wh{W} \neq W \right) \le \frac{1}{4} + \delta < \frac{1}{2}.
\end{align}
\end{lemma}

To finish the proof, we now apply the Meta-Theorem. For each $t$, let $\QQ_{Y_t|Z^{t-1}} = \QQ_{Y_t}$ be the normal distribution $N(0, \sigma^2 I_{n \times n})$. Then
\begin{align}
&	D\Big( \PP_{Y_t|Z^{t-1},W} \Big\| \QQ_{Y_t|Z^{t-1}}\Big| \PP_{Z^{t-1},W}\Big) \nonumber\\
& \qquad = \frac{1}{2\sigma^2}\E \| A_W Y_{t-1} + U_{t-1} \|^2 \\
& \qquad = \frac{1}{2\sigma^2}\E \| Y_t - V_t \|^2.
\end{align}
Then
\begin{align}
	&\frac{1}{2}\left(b_n + n^2 \log \frac{1}{4\sqrt{\eps/c}}\right) \nonumber\\
	&\quad\le \frac{1}{2\sigma^2}\sum^T_{t=1} \E \| Y_t - V_t \|^2 + \log 2 \\
	&\quad\le \frac{1}{2\sigma^2}\E \| Y_1 - V_1 \|^2 + \frac{1}{2\sigma^2} \sup_{A \in \Lambda}R_T(\gamma,A) + \log 2 \\
	&\quad\le \frac{C + n\sigma^2}{\sigma^2} + \log 2 + \frac{T\eps}{2\sigma^2}.
\end{align}
Rearranging, we obtain \eqref{eq:mintime}, and the theorem is proved. \end{proof}

\section{Conclusion}

We have presented a Meta-Theorem on the inevitable trade-offs between {\em a priori} uncertainty, {\em a posteriori} uncertainty, and the information accumulated online in the process of controlling an unknown stochastic dynamical system. The Meta-Theorem connects the notions of information, learning, and adaptation in the sense of Kolmogorov and Zames with the Shannon-theoretic notion of information gain quantified by the divergence between the actual sequence of the system kernels and some sequence of auxiliary stochastic kernels. The freedom of choosing these auxiliary kernels is what gives the Meta-Theorem its power. We have used the Meta-Theorem to derive fundamental lower bounds on the performance of system identification algorithms and on the minimum time needed to stabilize an uncertain linear system. As part of future work, we will investigate fundamental limits of robust estimation and control algorithms over uncertainty sets defined directly by divergence (relative entropy) constraints \cite{ChaRez07,SRC09}.

\section*{Acknowledgment} The author wishes to thank Tamer Ba\c{s}ar, Todd Coleman, Tara Javidi, Yury Polyanskiy, Cosma Shalizi, and Serdar Y\"uksel for stimulating discussions related to the  content of this work, as well as to its potential applications and extensions.

\begin{appendices}
	\renewcommand{\theequation}{\Roman{section}.\arabic{equation}}
	\setcounter{lemma}{0}
	\setcounter{equation}{0}

	\renewcommand{\thelemma}{\Roman{section}.\arabic{lemma}}
	
	\section{Proof of Lemma~\ref{lm:control_ident}}
	\label{app:ident_proof}
For brevity, we will write $F_t$ instead of $F_t(Y^t,U^{t-1})$. Suppose that the event in \eqref{eq:PEP} holds for a given $A \in \Lambda$. Then $G_T$ is invertible, and
\begin{align}
A - \tilde{A}_{T} &= \sum^{T}_{t=1} (A-F_t)Y_t Y_t^\tr G_T^{-1}.
\end{align}
Let $\Delta_t = A - F_t$ and $H_t = Y_t Y_t^\tr$. Then for any two vectors $u,v \in \Reals^n$ we have
\begin{align}
&\left|u^\tr (A - \tilde{A}_T) v\right|^2 \nonumber \\
&\le \left(\sum^{T}_{t=1} \left|u^\tr \Delta_t H_t G_T^{-1}v\right|\right)^2  \\
&\le \left( \sum^{T}_{t=1} \left\| \sqrt{H_t} \Delta_t^\tr u \right\| \, \left\| \sqrt{H_t} G^{-1}_T v \right\| \right)^2  \\
& \le \left(\sum^{T}_{t=1}  \left\| \sqrt{H_t} \Delta_t^\tr u \right\|^2 \right) \left(\sum^{T}_{t=1} \left\|  \sqrt{H_t} G^{-1}_T v \right\|^2 \right)  \\
& =\left(\sum^{T}_{t=1} u^\tr \Delta_t Y_t Y_t^\tr \Delta_t^\tr u \right) \cdot v^\tr G_T^{-1} v  \\
& \le \left(\sum^{T}_{t=1} u^\tr \Delta_t Y_t Y_t^\tr \Delta_t^\tr u \right) \cdot \frac{1}{cT} \| v \|^2,\label{eq:ident_bound}
\end{align}
where $\|\cdot\|$ denotes the Euclidean norm on $\Reals^n$, the third and the fourth steps use Cauchy--Schwarz, the fifth step uses the definition of $H_t$, and the last step uses the persistent excitation property. Taking the supremum of both sides of \eqref{eq:ident_bound} over all $v$ with $\| v \|=1$ and using the fact that
\begin{align}
\Delta_t Y_t = (A - F_t)Y_t = AY_t + U_t = Y_{t+1} - V_{t+1},
\end{align}
we obtain the bound
\begin{align}
\| (A - \tilde{A}_T) u \|^2 \le \frac{1}{cT} \sum^{T}_{t=1} \left|(Y_{t+1} - V_{t+1})^\tr u\right|^2
\end{align}
that holds for all $u \in \Reals^n$. Taking the supremum over all unit-norm $u$, we get the lemma.

\section{Proof of Lemma~\ref{lm:prob_error_upper}}
\label{app:upper_proof}

For every $i \in [N]$ define the following events:
\begin{align}
R^{(i)}_T &\deq \left\{ W=i \right\} \cap \left\{ \frac{1}{T} \sum^{T}_{t=1} \| Y_{t+1} - V_{t+1} \|^2 \ge 4\eps  \right\} \\
S^{(i)}_T & \deq \{ W = i \} \cap \left\{ \| \tilde{A}_T - A_i \| \ge 2 \sqrt{\eps/c} \right\} \\
E^{(i)}_T & \deq \{W = i \} \cap \left\{ \frac{G_T}{T} \succeq cI_{n\times n} \right\}.
\end{align}
Let $\PP_i(\cdot)$ and $\E_i\{\cdot\}$ denote $\PP(\cdot|W=i)$ and $\E\{\cdot|W=i\}$, respectively. If $\gamma$ achieves regret $\eps$ in time $T$, then by Markov's inequality
\begin{align}
\PP_i \Big(R^{(W)}_T\Big) &\le \frac{\E_i\left\{ \frac{1}{T}\sum^{T}_{t=1} \| Y_{t+1} - V_{t+1} \|^2\right\} }{4\eps} \le  \frac{1}{4}.
\end{align}
Now suppose that $W=i$, but $\wh{W} \neq i$ and $S^{(i)}_T$ is false. By definition of $\wh{W}$, we must then have
\begin{align}
	\big\| \td{A}_T - A_{\wh{W}} \big\| \le \big\| \td{A}_T - A_i \big\| < 2\sqrt{\eps/c}.
\end{align}
Moreover, since both $A_i$ and $A_{\wh{W}}$ belong to the $4\sqrt{\eps/c}$-packing set and $\wh{W}\neq i$, $\| A_i - A_{\wh{W}} \| \ge 4\sqrt{\eps/c}$. Then triangle inequality gives
\begin{align}
\big\| \tilde{A}_T - A_i \big\| \ge \big\| A_i - A_{\wh{W}} \big\| - \big\| A_{\wh{W}} - \tilde{A}_T \big\| > 2\sqrt{\eps/c}.
\end{align}
This contradicts the assumption that $S^{(i)}_T$ is false. Hence,
\begin{align}
\PP_i\big( \wh{W} \neq W \big) \le \PP_i\left( S^{(W)}_T \right).
\end{align}
By Lemma~\ref{lm:control_ident}, $S^{(i)}_T \cap E^{(i)}_T \subseteq R^{(i)}_T \cap E^{(i)}_T$. Therefore,
\begin{align}
& \PP_i\left( S^{(W)}_T \right) \nonumber\\
&= \PP_i\left(S^{(W)}_T \cap E^{(W)}_T\right) + \PP_i \left( S^{(W)}_T \cap \bar{E}^{(W)}_{T}\right) \\
 & \le \PP_i\left(R^{(W)}_T \cap E^{(W)}_T\right) + \PP_i \left( S^{(W)}_T \cap \bar{E}^{(W)}_{T}\right) \\
 & \le \PP_i\left(R^{(W)}_T\right) + \PP_i\left(\bar{E}^{(W)}_T\right) \\
 & \le \frac{1}{4} + \delta,
\end{align}
where the bar denotes set-theoretic complement. Averaging w.r.t.\ the distribution of $W$, we obtain the statement of the lemma.

\end{appendices}

\bibliographystyle{IEEEtran}
\bibliography{div_adaptive.bbl}
\end{document}